\documentclass[a4paper]{article}
\usepackage{natbib}
\usepackage[titletoc,toc,title]{appendix}
\usepackage{amsmath}
\usepackage{amssymb}
\usepackage{amsthm}
\usepackage{bbm}
\usepackage{algorithm}
\usepackage[noend]{algpseudocode}
\usepackage{hyperref}
\hypersetup{
    colorlinks=true,
    linkcolor=blue,
    filecolor=blue,      
    urlcolor=blue,
    citecolor= blue
}

\newtheorem{theorem}{Theorem}[section]
\newtheorem{corollary}{Corollary}[theorem]
\newtheorem{lemma}[theorem]{Lemma}

\allowdisplaybreaks

\title{\textbf{Sampled Fictitious Play is Hannan Consistent}}

\author{Zifan Li \\ {\tt zifanli@umich.edu}
\and Ambuj Tewari \\ {\tt tewaria@umich.edu} }

\begin{document}
\maketitle

\begin{abstract}
Fictitious play is a simple and widely studied adaptive heuristic for playing repeated games. It is well known that fictitious play fails to be Hannan consistent.
Several variants of fictitious play including regret matching, generalized regret matching and smooth fictitious play, are known to be Hannan consistent.
In this note, we consider sampled fictitious play: at each round, the player samples past times and plays the best response to previous moves of other players 
at the sampled time points. We show that sampled fictitious play, using Bernoulli sampling, is Hannan consistent. Unlike several existing Hannan consistency proofs that
rely on concentration of measure results, ours instead uses anti-concentration results from Littlewood-Offord theory.
\end{abstract}
%JEL classification: C73

\section{Introduction} \label{sec:intro}

In the setting of repeated games played in discrete time, the (unconditional) regret of a player, at any time point, is the difference between the payoffs she would have received had she played the best, in hindsight, constant strategy throughout, and
the payoffs she did in fact receive. \citet{Hannan1957} showed the existence of procedures with a ``no-regret'' property: procedures for which the average regret per time goes to zero for a large number of time points.
His procedure was a simple modification of fictitious play: random perturbations are
added to the cumulative payoffs of every strategy so far and the player picks the strategy with the largest perturbed cumulative payoff. No regret procedures are also called ``universally consistent" \cite[Section 4.7]{Fudenberg1998} or ``Hannan consistent" \citep[Section 4.2]{CesaBianchi2006}.

It is well known that smoothing the cumulative payoffs before computing the best response is crucial to achieve Hannan consistency. One way to achieve smoothness is through stochastic smoothing, or adding perturbations. Without perturbations, the procedure becomes identical to fictitious play, which fails to be Hannan consistent  \citep[Exercise 3.8]{CesaBianchi2006}. Besides Hannan's modification, other variants of fictitious play are also known to be Hannan consistent, including (unconditional) regret matching, generalized (unconditional) regret matching and smooth fictitious play (for an overview, see \citet[Section 10.9]{Hart2013}).

In this note, we consider another variant of fictitious play, namely sampled fictitious play. Here, the player samples past time points using some (randomized) sampling scheme and plays the best response to the moves of the other players restricted to the
set of sampled time points. Sampled fictitious play has been considered by other authors in different contexts. \cite{Kani1995} established convergence to Nash equilibrium in $2 \times 2$ games. \cite{Gilliland2006} provided regret bounds for the game of matching pennies. \cite{LambertIII2005} considered games with identical payoffs for all players and use sampled fictitious play to solve large-scale optimization problems. To the best of our knowledge, it is not known whether sampled fictitious play is Hannan consistent without making any assumptions on the form of the game and payoffs. The purpose of this note is to show that it is indeed Hannan consistent when used with a natural sampling scheme, namely Bernoulli sampling.

\section{Preliminaries} \label{sec:prelim}

Consider a game in strategic form where $M$ is the number of players, $S_i$ is the set of strategies for player $i$, and $u_i: \prod_{j=1}^M S_i \to \mathbb{R}$ is the payoff function for player $i$. For simplicity assume that the payoff functions of all players are $[-1,1]$ bounded. We also assume the number of pure strategies
is the same for each player and that $S_i = \{1,\ldots,N\}$. Let $S = \prod_{i=1}^M S_i$ be the set of $M$-tuples of player strategies. For $s = (s_i)_{i=1}^M \in S$, we denote the strategies of players other than $i$ by $s_{-i} = (s_j)_{1 \le j \le M, j\neq i}$.

The game is played repeatedly over (discrete) time $t = 1,2,\ldots$. A learning procedure for player $i$ is a procedure that maps the history $h_{t-1} = (s_\tau)_{\tau=1}^{t-1}$ of plays just prior to time $t$, to a strategy
$s_{t,i} \in S_i$. The learning procedure is allowed to be randomized, i.e., player $i$ has access to a stream of \ random variables $\epsilon_1,\epsilon_2,\ldots$ and she is allowed to use $\epsilon_1,\ldots,\epsilon_{t-1}$, in addition to
$h_{t-1}$, to choose $s_{t,i}$. Player $i$'s regret at time $t$ is defined as
\[
\mathcal{R}_{t,i} = \max_{k \in S_i} \sum_{\tau=1}^t u_i(k, s_{\tau,-i}) - \sum_{\tau=1}^t u_i(s_\tau) .
\]
This compares the player's cumulative payoff with the payoff she could have received had she selected the best constant (over time) strategy $k$ with knowledge of the other players' moves.

A learning procedure for player $i$ is said to be {\em Hannan consistent} if and only if
\[
\limsup_{t \to \infty} \frac{\mathcal{R}_{t,i}}{t} \le 0 \quad\quad \text{almost surely}.
\]
Hannan consistency is also known as the ``no-regret'' property and as ``universal consistency". The term ``universal" refers to the fact that the regret per time goes to zero irrespective of what the other players do.

{\em Fictitious play} is a (deterministic) learning procedure where player $i$ plays the best response to the plays of the other players so far. That is,
\begin{equation}
\label{eq:fp}
s_{t,i} \in \arg\max_{k \in \{1,\dots,N\}} \sum_{\tau=1}^{t-1} u_i(k, s_{\tau,-i}).
\end{equation}
As mentioned earlier, fictitious play is not Hannan consistent. However, consider the following modification of fictitious play, called {\em sampled fictitious play}. At time $t$, player randomly selects a subset $\mathbb{S}_t \subseteq \{1,\ldots,t-1\}$
of previous time points and plays the best response to the other players' moves only over $\mathbb{S}_t$. That is,
{\small
\begin{equation}
\label{eq:sfp}
s_{t,i} \in \arg\max_{k \in \{1,\dots,N\}} \sum_{\tau\in\mathbb{S}_t} u_i(k, s_{\tau,-i}).
\end{equation}
}If multiple strategies achieve the maximum, then the tie is broken uniformly at random, and independently with respect to all previous randomness. Also, if $\mathbb{S}_t$ turns out to be empty (an event that happens with probability exactly $2^{-(t-1)}$ under the Bernoulli sampling described below),  we adopt the convention that the argmax above includes all $N$ strategies.

In this note, we consider {\em Bernoulli sampling}, i.e., any particular round $\tau \in \{1,\ldots,t-1\}$ is included in $\mathbb{S}_t$ independently with probability $1/2$ . More specifically, if $\epsilon^{(t)}_1,\ldots,\epsilon^{(t)}_{t-1}$ are i.i.d.\ symmetric Bernoulli (or Rademacher) random variables taking values in $\{-1,+1\}$, then
\begin{equation}\label{eq:bs}
\mathbb{S}_t = \{ \tau \in \{1,\ldots,t-1\} \::\: \epsilon^{(t)}_\tau = +1 \}
\end{equation}
and therefore,
\[
\sum_{\tau\in\mathbb{S}_t} u_i(k, s_{\tau,-i}) = \sum_{\tau = 1}^{t-1} \frac{(1+\epsilon^{(t)}_\tau)}{2}  u_i(k, s_{\tau,-i}).
\]
Note that the procedure defined by the combination of~\eqref{eq:sfp} and~\eqref{eq:bs} is completely parameter free, i.e., there is no tuning parameter that has to be carefully tuned in order to obtain desired convergence properties.

\section{Result and Discussion}\label{sec:discussion}

Our main result is the following.

\begin{theorem}\label{thm:main}
Sampled fictitious play~\eqref{eq:sfp} with Bernoulli sampling~\eqref{eq:bs} is Hannan consistent.
\end{theorem}

Before we move on to the proof, a few remarks are in order.

\paragraph{Computational tractability} It is a simple but important observation that the \emph{form} of the optimization problem solved by fictitious play~\eqref{eq:fp} is \emph{exactly} the same as the optimization problem solved
by sampled fictitious play~\eqref{eq:sfp}. This can be very useful when the player has a large strategy set and does not want to enumerate all strategies to solve the maximization involved in both fictitious play and its sampled version.
For example, \cite{LambertIII2005} describe their computational experience with sampled fictitious play in the context of a dynamic traffic assignment problem.

\paragraph{Rate of convergence}
Our proof gives the rate of convergence of (expected) average regret as $O(N^2 \sqrt{ \log \log t / t})$ where the constant hidden in $O(\cdot)$ notation is small and explicit. It is known that the optimal rate is $O(\sqrt{\log N/t})$ \cite[Section 2.10]{CesaBianchi2006}. Therefore, our rate of convergence is almost optimal in $t$ but severely suboptimal in $N$.  This raises several interesting questions. What is the best bound possible for Sampled Fictitious Play with Bernoulli sampling? Is there a sampling scheme for which Sampled Fictitious Play procedure achieves the optimal rate of convergence? The first question is partially answered by Theorem $\ref{thm:counterexample}$ in Appendix \ref{sec:counterexample} which states that the dependency on $N$ is likely to be polynomial instead of logarithmical, but there is still some gap between the lower bound and the upper bound we provide.

\paragraph{Asymmetric probabilities}
Instead of using symmetric Bernoulli probabilities, we can choose $\epsilon^{(t)}_\tau$ such that $P(\epsilon^{(t)}_\tau = +1) = \alpha$. As $\alpha \to 1$, the learning procedure becomes fictitious play and as $\alpha \to 0$, it selects strategies uniformly at random. Therefore, it is natural to expect that the regret bound will blow up near the two extremes of $\alpha = 1$ and $\alpha = 0$. We can make this intuition precise, but only for $\{-1,0,1\}$-valued payoffs (instead of $[-1,1]$-valued). For details, see Appendix~\ref{sec:AB} in the supplementary material.

\paragraph{Follow the perturbed leader} Note that
\[
\arg\max_{k \in \{1,\dots,N\}} \sum_{\tau = 1}^{t-1} \frac{(1+\epsilon^{(t)}_\tau)}{2}  u_i(k, s_{\tau,-i})
= \arg\max_{k \in \{1,\dots,N\}}  \Big( \sum_{\tau = 1}^{t-1} u_i(k, s_{\tau,-i}) +  \sum_{\tau = 1}^{t-1} \epsilon^{(t)}_\tau u_i(k, s_{\tau,-i}) \Big) .
\] Therefore, we can think of sampled fictitious play as adding a random perturbation to the expression that fictitious play optimizes. Such algorithms are referred to as ``follow the perturbed leader'' (FPL) in the computer science literature (``fictitious play'' is known as ``follow the leader''). This family was originally proposed by \cite{Hannan1957} and popularized by \citet{Kalai2005}. Closer to this paper are the FPL algorithms of \citet{Devroye2013} and \citet{Erven2014}. However, none of these papers considered sampled fictitious play.

\paragraph{Extension to conditional (or internal) regret}
In this paper we focus on unconditional (or external) regret. Other notions of regret, especially conditional (or internal) regret can also be considered. Internal regret measures
the worst regret, over $N(N-1)$ choices of $k \neq k'$, of the form ``every time strategy $k$ was picked, strategy $k'$ should have been picked instead''. 
There are generic conversions \citep{Stoltz2005,Blum2007} that will convert any learning procedure with small external regret to one with small internal regret. These conversion, however, require access to the probability distribution over strategies at each time point. This probability distribution can be approximated, to arbitrary accuracy, by making the choice of the strategy in~\eqref{eq:sfp} multiple times each time selecting the random subset $\mathbb{S}_t$ independently. However, doing so and using a generic conversion from external to internal regret will lead to a cumbersome overall algorithm. It will be nicer to design a simpler sampling based learning procedure with small internal regret.

\section{Proof of the Main Result}\label{sec:proof}

We break the proof of our main result into several steps. The first and third steps involve fairly standard arguments in this area. Our main innovations are in step two.

\subsection{Step 1: From Regret to Switching Probabilities}\label{subsec:1}

In this step, we assume that players other than player $i$ (the ``opponents") are {\em oblivious}, i.e., they do not adapt to what player $i$ does. Mathematically, this means that the sequence $s_{t,-i}$ does not depend on the moves $s_{t,i}$ of player $i$. We will prove a uniform regret bound that holds for all deterministic payoff sequences $\{s_{t,-i} \}^T_{t=1}$, by which we can conclude that the same bound holds for oblivious but random payoff sequences as well. Since player $i$ is fixed for the rest of the proof, we will not carry the index $i$ in our notation further. 
Let the vector $g_t \in [-1,1]^{N}$ be defined as $g_{t,k} = u_i(k,s_{t,-i})$ for $k \in \{1,\ldots,N\}$. Moreover, we denote player $i$'s move $s_{t,i}$ as $k_t$. With this notation, regret at time $T$ equals
\[
\mathcal{R}_T = \max_{k \in \{1,\dots,N\}} \sum_{t=1}^T g_{t,k} - \sum_{t=1}^T g_{t,k_t} .
\]
In this step, we will look at the expected regret. Because the opponents are oblivious, this equals
\[
\mathbb{E}\left[ \mathcal{R}_T \right] = \max_{k \in \{1,\dots,N\}} \sum_{t=1}^T g_{t,k} - \mathbb{E}\left[ \sum_{t=1}^T g_{t,k_t} \right]
= \max_{k \in \{1,\dots,N\}} \sum_{t=1}^T g_{t,k} -  \sum_{t=1}^T \mathbb{E}\left[ g_{t,k_t} \right] .
\]
Recall that 
\[
k_t \in \arg\max_{k \in \{1,\dots,N\}}  \sum_{\tau = 1}^{t-1} \frac{(1+\epsilon^{(t)}_\tau)}{2} g_{\tau,k}.
\]
Since $g_t$'s are fixed vectors, by independence we see that the distribution of $k_t$ is exactly the same whether or not we share the Rademacher random variables across time points. Therefore, we do not have to draw a fresh sample $\epsilon^{(t)}_1,\ldots,\epsilon^{(t)}_{t-1}$ at time $t$. Instead, we fix a single stream $\epsilon_1,\epsilon_2,\ldots$ of i.i.d.\ Rademacher random variables and set
$(\epsilon^{(t)}_1,\ldots,\epsilon^{(t)}_{t-1}) = (\epsilon_1,\ldots,\epsilon_{t-1})$ for all $t$.  With this reduction in number of random variables used, we now have
{\small
\begin{equation}\label{eq:sfp_single}
k_t \in \arg\max_{k \in \{1,\dots,N\}} \sum_{\tau=1}^{t-1} (1+\epsilon_\tau) g_{\tau,k}.
\end{equation}
}

We define $G_t = \sum_{\tau=1}^{t}g_\tau$, the cumulative payoff vector at time $t$. Define $\tilde{g}_t = (1+\epsilon_t) g_t$ and $\tilde{G}_t = \sum_{\tau=1}^{t}\tilde{g}_\tau$. We also define 
$$ 
g_{t,i \ominus j} = g_{t,i} - g_{t,j}\ ,\ \tilde{g}_{t,i \ominus j} = \tilde{g}_{t,i} - \tilde{g}_{t,j} .
$$
With these definitions, we have
\begin{multline*}
\tilde{G}_{t,i \ominus j} = \tilde{G}_{t,i} - \tilde{G}_{t,j} 
= \sum_{\tau=1}^{t} \tilde{g}_{\tau,i} - \sum_{\tau=1}^{t} \tilde{g}_{\tau,j} 
\\
= \sum_{\tau=1}^{t} (1+\epsilon_\tau)(g_{\tau,i} - g_{\tau,j}) 
= \sum_{\tau=1}^{t} (1+\epsilon_\tau)g_{\tau,i \ominus j} .
\end{multline*}
The following result upper bounds the regret in terms of downward zero-crossings of the process $\tilde{G}_{t,i \ominus j}$, i.e., the times $t$ when it switches from
being non-negative at time $t-1$ to non-positive at time $t$.

\begin{theorem}
\label{thm:regret2switches}
We have the following upper bound on the expected regret:
$$
\mathbb{E}\left[ \mathcal{R}_T \right] \leq2 N^2 \max\limits_{1 \leq i,j \leq N} \sum_{t=1}^{T}  |g_{t,i \ominus j}|P\left(  \tilde{G}_{t-1,i \ominus j} \geq 0, \tilde{G}_{t,i \ominus j} \le 0 \right) .
$$
\end{theorem}
The proof of this theorem can be found in Appendix \ref{sec:proofs}. We now focus on bounding the switching probabilities for a fixed pair $i,j$.

\subsection{Step 2: Bounding Switching Probabilities Using Littlewood-Offord Theory}

Our strategy is to do a ``multi-scale'' analysis and, within each scale, apply Littlewood-Offord theory to bound the switching probabilities. The need for a multi-scale argument arises from the requirement in Littlewood-Offord theorem
(see Theorem~\ref{thm:LO} below) for a lower bound on the step sizes of random walks. We partition the set of $T$ time points $[T] := \{1,\ldots,T\}$ into $K+1$ disjoint sets at different scales, denoted as $\{A_k\}_{k=0}^{K}$ where 
$$
A_k = \begin{cases}
\{t \in [T] : |g_{t,i \ominus j}| \leq \frac{1}{\sqrt{T}}\} \quad &k = 0\\
\{t \in [T] : T^{-\frac{1}{2^k}} < |g_{t,i \ominus j}| \leq T^{-\frac{1}{2^{k+1}}}\} \quad &k = 1, \dots, K-1\\
\{t \in [T] : T^{-\frac{1}{2^K}} <|g_{t,i \ominus j}| \leq 2\} \quad &k = K
\end{cases}.
$$
Note that actually $A_k$ depends on $i,j$ as well but for the sake of clarity we drop this dependence in the notation. The cardinality of a finite set $A$ will be denoted by $|A|$. The number $K+1$ of different scales is determined by
$$
K = \arg \min \{k \in \mathbb{N} \::\: T^{-\frac{1}{2^k}} \geq 1/2 \} .
$$ 
$\forall t,i, \ g_{t,i} \in [-1,1]$ so $|g_{t,i \ominus j}| \in [0,2]$. The scales here are chosen such that $K$ is not very large (of order $O\Big(\log \log(T)\Big)$) and still covers the entire range of the payoffs.
It easily follows that,
\begin{align*}
&\sum_{t=1}^{T}  |g_{t,i \ominus j}|P\left(  \tilde{G}_{t-1,i \ominus j} \geq 0, \tilde{G}_{t,i \ominus j} \le 0 \right) \\
= & \sum_{t=1}^T |g_{t,i \ominus j}|P\left(\sum_{\tau=1}^{t-1}\epsilon_\tau g_{\tau,i \ominus j} \geq -\sum_{\tau=1}^{t-1} g_{\tau,i \ominus j} , \sum_{\tau=1}^{t}\epsilon_\tau g_{\tau,i \ominus j} \le -\sum_{\tau=1}^{t} g_{\tau,i \ominus j}\right) \\
= & \sum_{k=0}^{K} \sum_{t \in A_k} |g_{t,i \ominus j}|P\left(\sum_{\tau=1}^{t-1}\epsilon_\tau g_{\tau,i \ominus j} \geq -\sum_{\tau=1}^{t-1} g_{\tau,i \ominus j} , \sum_{\tau=1}^{t}\epsilon_\tau g_{\tau,i \ominus j} \le -\sum_{\tau=1}^{t} g_{\tau,i \ominus j}\right) .
\end{align*}
We now want to argue that the probabilities involved above are small. The crucial observation is that, if a switch occurs, then the random sum $\sum_{\tau=1}^{t}\epsilon_\tau g_{\tau,i \ominus j}$
has to lie in a sufficiently small interval. Such ``small ball'' probabilities are exactly what the classic Littlewood-Offord theorem controls.

\begin{theorem}[Littlewood-Offord Theorem of Erd\"os, Theorem 3 of \cite{Erdos1945}]
\label{thm:LO}
Let $x_1, \dots, x_n$ be n real numbers such that $|x_i| \geq 1$ for all $i$. For any given radius $\Delta > 0$, the small ball probability satisfies
$$
\sup\limits_{B}P(\epsilon_1x_1+\dots+\epsilon_nx_n \in B) \leq \frac{S(n)}{2^n}(\lfloor \Delta \rfloor + 1) 
$$ 
where $\epsilon_1, \dots, \epsilon_n$ are i.i.d.\ Rademacher random variables, $B$ ranges over all closed balls (intervals) of radius $\Delta$, and $\lfloor x \rfloor$ refers to the integral part of $x$, $S(n)$ is the largest binomial coefficient belonging to n.
\end{theorem}

Using elementary calculations to upper bound $\frac{S(n)}{2^n}$ gives us the following corollary.

\begin{corollary}
\label{cor:LO}
Under the same notation and conditions as Theorem~\ref{thm:LO}, we have 
$$
\sup\limits_{B}P(\epsilon_1x_1+\dots+\epsilon_nx_n \in B) \leq C_{LO}(\lfloor \Delta \rfloor + 1)\frac{1}{\sqrt{n}}
$$ 
where $C_{LO} = \frac{2\sqrt{2}e}{\pi} < 3$.
\end{corollary}
The proof of this corollary can be found in Appendix \ref{sec:proofs}.

The scale of payoffs for time periods in $A_0$ is so small that we do not need any Littlewood-Offord theory to control their contribution to the regret. Simply bounding the probabilities by $1$ gives us the following.

\begin{theorem}\label{thm:A_0}
The following upper bound holds for switching probabilities for time periods within $A_0$:
\begin{multline*}
\sum_{t \in A_0} |g_{t,i \ominus j}|P\left(\sum_{\tau=1}^{t-1}\epsilon_\tau g_{\tau,i \ominus j} \geq -\sum_{\tau=1}^{t-1} g_{\tau,i \ominus j} , \sum_{\tau=1}^{t}\epsilon_\tau g_{\tau,i \ominus j} \le -\sum_{\tau=1}^{t} g_{\tau,i \ominus j}\right) \\
\leq {\sqrt{|A_0|}} \leq 20C_{LO} {\sqrt{|A_0|}}.
\end{multline*}
where $C_{LO} > 1$.
\end{theorem}
The proof of this theorem can also be found in Appendix \ref{sec:proofs}.

The real work lies in controlling the switching probabilities for payoffs at intermediate scales. The idea in the proof of the results is to condition on the $\epsilon_t$'s outside $A_k$. Then the probability of interest
is written as a small ball event in terms of the $\epsilon_t$'s in $A_k$. Applying Littlewood-Offord theorem concludes the argument.

\begin{theorem}
\label{thm:A_intermediate}
For any $k \in \{1, \dots, K\}$, we have
\begin{multline*}
\sum_{t \in A_k} |g_{t,i \ominus j}|P\left(\sum_{\tau=1}^{t-1}\epsilon_\tau g_{\tau,i \ominus j} \geq -\sum_{\tau=1}^{t-1} g_{\tau,i \ominus j} , \sum_{\tau=1}^{t}\epsilon_\tau g_{\tau,i \ominus j} \le -\sum_{\tau=1}^{t} g_{\tau,i \ominus j}\right) \\
\leq 20C_{LO}\sqrt{|A_k|} .
\end{multline*}
\end{theorem}
Again, the proof of this theorem is deferred to Appendix \ref{sec:proofs}.

We finally have all the ingredients in place to control the switching probabilities.

\begin{corollary}
\label{cor:switches}
The following upper bound on the switching probabilities holds.
\begin{multline*}
\sum_{t=1}^T |g_{t,i \ominus j}|P\left(\sum_{\tau=1}^{t-1}\epsilon_\tau g_{\tau,i \ominus j} \geq -\sum_{\tau=1}^{t-1} g_{\tau,i \ominus j} , \sum_{\tau=1}^{t}\epsilon_\tau g_{\tau,i \ominus j} \le -\sum_{\tau=1}^{t} g_{\tau,i \ominus j}\right) \\
\le 20C_{LO} \sqrt{T \log_2 (4 \log_2 T)}.
\end{multline*}
\end{corollary}
\begin{proof}
Using Theorem~\ref{thm:A_0} and Theorem~\ref{thm:A_intermediate}, we have
\begin{align*}
&\sum_{t=1}^T |g_{t,i \ominus j}|P\left(\sum_{\tau=1}^{t-1}\epsilon_\tau g_{\tau,i \ominus j} \geq -\sum_{\tau=1}^{t-1} g_{\tau,i \ominus j} , \sum_{\tau=1}^{t}\epsilon_\tau g_{\tau,i \ominus j} \le -\sum_{\tau=1}^{t} g_{\tau,i \ominus j}\right) \\
= & \sum_{k=0}^{K} \sum_{t \in A_k} |g_{t,i \ominus j}|P\left(\sum_{\tau=1}^{t-1}\epsilon_\tau g_{\tau,i \ominus j} \geq -\sum_{\tau=1}^{t-1} g_{\tau,i \ominus j} , \sum_{\tau=1}^{t}\epsilon_\tau g_{\tau,i \ominus j} \le -\sum_{\tau=1}^{t} g_{\tau,i \ominus j}\right) \\
\leq & \sum_{k=0}^K 20C_{LO} \sqrt{|A_k|} .
\end{align*}
Since $\sum_{k=0}^K \sqrt{|A_k|} \le \sqrt{K+1} \cdot \sqrt{ \sum_{k=0}^K |A_k| }$ and
$\sum_{k=0}^K |A_k| = T$, we have
$$
\sum_{k=0}^K 20C_{LO} \sqrt{|A_k|} \leq 20C_{LO} \sqrt{(K+1)T}.
$$
By definition of $K$, we have that $T^{-\frac{1}{2^{K-1}}} < \frac{1}{2}$, $K \leq \log_2(\log_2(T)) + 1$ which finishes the proof.
\end{proof}

Thus, $\forall i, j \in \{1, \dots, N, i \neq j\}$, we have
\begin{align*}
&\sum_{t=1}^T |g_{t,i \ominus j}|P\left(\tilde{G}_{t-1,i \ominus j} \geq 0, \tilde{G}_{t,i \ominus j} \le 0 \right) \le 20C_{LO} \sqrt{T \log_2 (4 \log_2 T)},
\end{align*}
which, when plugged into Theorem~\ref{thm:regret2switches}, immediately yields the following corollary.

\begin{corollary}
\label{cor:obliviousregret}
Against an oblivious opponent, both versions --- the single stream version~\eqref{eq:sfp_single} and the fresh-randomization-at-each-round version~\eqref{eq:sfp} --- of sampled fictitious play
enjoy the following bound on the expected regret.
$$
\mathbb{E}\left[ \mathcal{R}_T \right] \le 40 C_{LO} N^2 \sqrt{T \log_2 (4 \log_2 T)} .
$$
\end{corollary}

\subsection{Step 3: From Oblivious to Adaptive Opponents}

Now we consider adaptive opponents. In this setting, we can no longer assume that player $i$ plays against a fixed sequence of payoff vectors $\{g_t\}_{t=1}^T$. Note that $g_{t,k}$ is just shorthand for $u_i(k,s_{t,-i})$ and opponents can react to player $i$'s moves $k_1,\ldots,k_{t-1}$ in selecting their strategy tuple $s_{t,-i}$, possibly making use of their own private randomness. We denote all randomness used collectively by other players over all time periods by $\omega$ which is drawn from some probability space $\Omega$. Thus, $g_t$ is a function $g_t(k_1,\ldots,k_{t-1},\omega)$.
Faced with general adaptive opponents, the single stream version~\eqref{eq:sfp_single} can incur terrible expected regret as stated below.

\begin{theorem}
\label{thm:bad_regret}
The single stream version of the sampled fictitious play procedure~\eqref{eq:sfp_single} can incur linear expected regret against adaptive opponents.
\end{theorem}
The proof of this theorem can be found at the end of Appendix \ref{sec:proofs}.

However, for the fresh randomness at each round procedure~\eqref{eq:sfp}, we can apply Lemma 4.1 of \citet{CesaBianchi2006} along with Corollary~\ref{cor:obliviousregret} to derive our next result that holds for
adaptive opponents too. There are two conditions that we must verify before we apply that lemma. First, the learning procedure should use independent randomization at different time points. Second, the probability distribution
of $s_{t,i}$ over the $N$ available strategies should be fully determined by $s_{1,-i},\ldots,s_{t-1,-i}$ and should not depend explicitly on player $i$ own previous moves $s_{1,i},\ldots,s_{t-1,i}$. Both of these
conditions are easily seen to hold for sampled fictitious play as defined in~\eqref{eq:sfp} and~\eqref{eq:bs}. Also note that \citet{CesaBianchi2006} consider deterministic adaptive opponents in their Lemma 4.1.
The extension to our case is easy: we first get a high probability (w.r.t. player $i$'s randomness) regret bound for the deterministic adaptive opponent $g_t(k_1,\ldots,k_{t-1},\omega)$ for a fixed $\omega$. Since the bound
holds for every $\omega$ and does not depend on $\omega$, the same high probability bound is true when $\omega$ is drawn from $\Omega$. This leads us to our final result.

\begin{theorem}
For any T, for any $\delta_T > 0$, with probability at least $1-\delta_T$, the actual regret $\mathcal{R}_T$ of sampled fictitious play as defined in~\eqref{eq:sfp} and~\eqref{eq:bs} satisfies, for any adaptive opponent,
$$
\mathcal{R}_T  \leq 40 C_{LO} N^2\sqrt{T \log_2(4\log_2 T)} + \sqrt{\frac{T}{2}\log\frac{1}{\delta_T}} .
$$
\end{theorem}

Now pick $\delta_T = \frac{1}{T^2}$. Consider the events $E_T  = \{ \mathcal{R}_T \ge 12 C_{LO} N^2\sqrt{T \log_2(4\log_2 T)} + \sqrt{T \log T}\}$ with $P(E_T) \le \delta_T$. Since
$\sum_{T=1}^\infty \delta_T < \infty$, we have $\sum_{T=1}^\infty P(E_T) < \infty$. Therefore, using Borel-Cantelli lemma, the event ``infinitely many $E_T$'s occur'' has probability $0$. That is, with probability $1$, we have
$\limsup_{T\to \infty} \frac{\mathcal{R}_T}{T \log T} \le C$ for some constant C. In particular, with probability $1$, $\limsup_{T \to \infty} \frac{\mathcal{R}_T}{T} = 0$, which proves Theorem~\ref{thm:main}.

\section{Conclusion}\label{sec:con}

We proved that a natural variant of fictitious play is Hannan consistent. In the variant we considered, the player plays the best response to moves of her opponents at sampled time points in the history so far.
We considered one particular sampling scheme, namely Bernoulli sampling. It will be interesting to consider other sampling strategies including sampling with replacement. It will also be interesting to consider notions of regret, such as tracking regret \citep[Section 5.2]{CesaBianchi2006}, that are more suitable for non-stationary environments by biasing the sampling to give more importance to recent time points.

\section*{Acknowledgements}

We thank Jacob Abernethy, Gergely Neu and Manfred Warmuth for helpful discussions. 
We acknowledge the support of NSF via CAREER grant IIS-1452099.

\nocite{*}
\bibliography{sampled_fp}

\newpage
\begin{appendices}

\section{Proofs}\label{sec:proofs}
We first present a lemma that helps us in proving Theorem \ref{thm:regret2switches}.
\begin{lemma}\label{lem:btl}
Let $k_t$ and $\tilde{g}_t$ be defined as in~\eqref{eq:sfp_single} and the text following that equation. We have,
$$
\sum_{t=1}^T \tilde{g}_{t,k_{t+1}} \geq \sum_{t=1}^T \tilde{g}_{t,k_{T+1}} = \max_{k \in \{1,\dots,N\}} \sum_{t=1}^T \tilde{g}_{t,k} .
$$
\end{lemma}
\begin{proof}
This is a classical lemma, for example, see Lemma 3.1 in \citep{CesaBianchi2006}. We follow the same idea, i.e, proving through induction but adapt it to handle gains instead of losses. The statement is obvious for $T = 1$. Assume now that 
$$
\sum_{t=1}^{T-1} \tilde{g}_{t,k_{t+1}} \geq  \sum_{t=1}^{T-1} \tilde{g}_{t,k_{T}}.
$$
Since, by definition, $\sum_{t=1}^{T-1} \tilde{g}_{t,k_{T}} \geq \sum_{t=1}^{T-1} \tilde{g}_{t,k_{T+1}}$, the inductive assumption implies 
$$
\sum_{t=1}^{T-1} \tilde{g}_{t,k_{t+1}} \geq \sum_{t=1}^{T-1} \tilde{g}_{t,k_{T+1}}.
$$
Add $\tilde{g}_{T,k_{T+1}}$ to both sides to obtain the result.
\end{proof}

\begin{proof}[Proof of Theorem \ref{thm:regret2switches}]
We will prove a result for Bernoulli sampling with general probabilities, i.e., when $P(\epsilon_t = +1) = \alpha$ where $\alpha$ is not necessarily $1/2$. We will show that
$$
\mathbb{E}\left[ \mathcal{R}_T \right] \leq \frac{N^2}{\alpha} \max\limits_{1 \leq i,j \leq N} \sum_{t=1}^{T}  |g_{t,i \ominus j}|P\left(  \tilde{G}_{t-1,i \ominus j} \geq 0, \tilde{G}_{t,i \ominus j} \le 0 \right)
$$
from which the theorem follows as a special case when $\alpha = 1/2$.

Obviously we have $\mathbb{E}(\tilde{g}_{t,i}) = 2\alpha g_{t,i} 
$ because of the fact that $\mathbb{E}(\epsilon_t) = 2\alpha - 1$. Furthermore, $\mathbb{E}[\tilde{g}_{t,k_t}|\epsilon_1, \dots, \epsilon_{t-1}] =  2\alpha g_{t,k_t}$ because $k_t$ is fully determined by past randomness $\epsilon_1, \dots, \epsilon_{t-1}$ and past payoffs $g_1, \dots, g_{t-1}$ that are given. This implies that $ \mathbb{E}[\tilde{g}_{t,k_t}] = \mathbb{E}\left[\mathbb{E}[\tilde{g}_{t,k_t}|\epsilon_1, \dots, \epsilon_{t-1}]\right]=  2\alpha \mathbb{E}[g_{t,k_t}]$. We now have,
\begin{align*}
\mathbb{E}\left[ \mathcal{R}_T \right] &= \max_{k \in \{1,\dots,N\}} \sum_{t=1}^T g_{t,k} - \mathbb{E}\left[ \sum_{t=1}^T g_{t,k_t} \right] \\ 
&=  \frac{1}{2\alpha} \max_{k \in \{1,\dots,N\}} \mathbb{E}\left[\sum_{t=1}^T \tilde{g}_{t,k}\right] - \frac{1}{2\alpha}\mathbb{E}\left[\sum_{t=1}^T \tilde{g}_{t,k_t}\right] \\
&\leq \frac{1}{2\alpha} \mathbb{E}\left[\max_{k \in \{1,\dots,N\}} \sum_{t=1}^T \tilde{g}_{t,k} - \sum_{t=1}^T \tilde{g}_{t,k_t} \right] .
\end{align*}
Using Lemma~\ref{lem:btl}, we can further upper bound the last expression as follows,
\begin{align*}
\mathbb{E}\left[ \mathcal{R}_T \right] &\leq \frac{1}{2\alpha} \mathbb{E}\left[\sum_{t=1}^T \tilde{g}_{t,k_{t+1}} - \sum_{t=1}^T \tilde{g}_{t,k_{t}}\right] \\
&= \frac{1}{2\alpha} \sum_{t=1}^T \mathbb{E}\left[(1+\epsilon_{t})(g_{t,k_{t+1}} - g_{t,k_{t}})\right] \\
&\leq \frac{1}{2\alpha} \sum_{t=1}^T \mathbb{E}\left[(1+\epsilon_t)|g_{t,k_{t+1}} - g_{t,k_{t}}|\right] \\
&\leq \frac{1}{\alpha} \sum_{t=1}^T \mathbb{E}\left[|g_{t,k_t} - g_{t,k_{t+1}}|\right] \\
&= \frac{1}{\alpha} \sum_{t=1}^T \sum\limits_{1 \leq i,j \leq N} \mathbb{E}\left[|g_{t,i} - g_{t,j}| 1_{(k_t =i, k_{t+1}=j)}\right] \\
&= \frac{1}{\alpha}  \sum\limits_{1 \leq i,j \leq N} \sum_{t=1}^T \mathbb{E}\left[|g_{t,i} - g_{t,j}| 1_{(k_t =i, k_{t+1}=j)}\right] \\
&\leq \frac{N^2}{\alpha} \max\limits_{1 \leq i,j \leq N} \sum_{t=1}^{T} |g_{t,i} - g_{t,j}|P(k_t = i, k_{t+1} = j) \\ 
&\leq \frac{N^2}{\alpha} \max\limits_{1 \leq i,j \leq N} \sum_{t=1}^{T} |g_{t,i} - g_{t,j}|P\left(\tilde{G}_{t-1,i} \geq \tilde{G}_{t-1,j}, \tilde{G}_{t,i} \le \tilde{G}_{t,j}\right) \\
&= \frac{N^2}{\alpha} \max\limits_{1 \leq i,j \leq N} \sum_{t=1}^{T}  |g_{t,i \ominus j}|P\left(  \tilde{G}_{t-1,i \ominus j} \geq 0, \tilde{G}_{t,i \ominus j} \le 0 \right).
\end{align*}
\end{proof}

The next lemma is useful to determine the appropriate constant in the Littlewood-Offord Theorem.
\begin{lemma}
\label{lemma:bernoullicoeff}
Suppose $X_1,\dots,X_t$ are i.i.d.\ Bernoulli random variables that take value of $1$ with probability $\alpha$ and $0$ with probability $1 - \alpha$. If $t > \max(\frac{2}{1-\alpha},\frac{2}{\alpha}) \geq \max(\frac{2\alpha}{1-\alpha},\frac{2}{\alpha}) $,
then for all $k$,
$$
P\left(\sum_{i=1}^t X_i = k \right) \leq \frac{e}{2\pi}\times \sqrt{\frac{2}{\alpha(1-\alpha)}} \times t^{-\frac{1}{2}}  .
$$
\end{lemma}
\begin{proof}
Note that for $0 \leq k < t$,
\[
\frac{P(x=k+1)}{P(x=k)} = \frac{{t \choose k+1} \alpha^{k+1} (1-\alpha)^{t-k-1}}{{t \choose k} \alpha^{k} (1-\alpha)^{t-k}} = \frac{\alpha (t-k) }{(1-\alpha)(k+1)}.
\]
Therefore, the maximum probability of Bernoulli distribution $P(X = k)$ is achieved when $k = \hat{k} = \lfloor(t+1)\alpha \rfloor$ where $\lfloor x \rfloor$ denotes the integral part of x. Clearly $\hat{k} \in [t\alpha-1, (t+1)\alpha]$. Thus, 
\begin{align*}
\sqrt{\hat{k}(t-\hat{k})} &\geq \min\Big(\sqrt{(t\alpha-1)(t-t\alpha+1)}, \sqrt{(t+1)\alpha(t-t\alpha-\alpha)}\Big) \\
&= t \times \min\Big(\sqrt{(\alpha-\frac{1}{t})(1-\alpha+\frac{1}{t})}, \sqrt{(1+\frac{1}{t})\alpha(1-\alpha-\frac{\alpha}{t})}\Big) \\
&\geq t \times \min\Big(\sqrt{(\alpha-\frac{\alpha}{2})(1-\alpha)}, \sqrt{\alpha(1-\alpha-\frac{1-\alpha}{2})}\Big) \\
&= \sqrt{\frac{\alpha(1-\alpha)}{2}}t.
\end{align*}
With this preliminary inequality, we are ready to prove the lemma.
\begin{align*} 
P(\sum_{i=1}^{t}X_i = k) &\leq P(\sum_{i=1}^{t}X_i = \hat{k}) \\
&= {t \choose \hat{k}} \times \alpha^{\hat{k}}(1-\alpha)^{t-\hat{k}} \\
&= \frac{t!}{(\hat{k})!(t-\hat{k})!} \times \alpha^{\hat{k}}(1-\alpha)^{t-\hat{k}}\\
&\leq \frac{t^{t+\frac{1}{2}}e^{1-t}}{\big(\sqrt{2\pi}{\hat{k}}^{\hat{k} + \frac{1}{2}}e^{-\hat{k}}\big)\big(\sqrt{2\pi}{(t-\hat{k})}^{t-\hat{k} + \frac{1}{2}}e^{-(t-\hat{k})}\big)} \times \alpha^{\hat{k}}(1-\alpha)^{t-\hat{k}} \\
&= \frac{e}{2\pi} \times \frac{1}{\sqrt{\hat{k}(t-\hat{k})}} \times \frac{t^{t+\frac{1}{2}}}{\hat{k}^{\hat{k}}(t-\hat{k})^{t-\hat{k}}} \times \alpha^{\hat{k}}(1-\alpha)^{t-\hat{k}} \\
&\leq \frac{e}{2\pi} \times \sqrt{\frac{2}{\alpha(1-\alpha)}}\times t^{t-\frac{1}{2}} \times \frac{\alpha^{\hat{k}}(1-\alpha)^{t-\hat{k}}}{\hat{k}^{\hat{k}}(t-\hat{k})^{t-\hat{k}}}.
\end{align*}
Let $f(x) = \frac{\alpha^{x}(1-\alpha)^{t-x}}{x^{x}(t-x)^{t-x}}, f'(x) = \Big(\log(\frac{\alpha}{1-\alpha}) - \log(\frac{x}{t-x})\Big) \times f(x)
$. Obviously $f'(x)$ is 0 when $x = \alpha t$, positive when $x < \alpha t$, and negative when $x > \alpha t$. Thus, 
$$
f(x) \leq \frac{\alpha^{\alpha t}(1-\alpha)^{t-\alpha t}}{(\alpha t)^{\alpha t}(t-\alpha t)^{t-\alpha t}} = t^{-t}.
$$
Hence,
\begin{align*}
P(\sum_{i=1}^{t}X_i = a) &\leq \frac{e}{2\pi} \times \sqrt{\frac{2}{\alpha(1-\alpha)}} \times t^{t-\frac{1}{2}} \times f(\hat{k}) \\
&\leq \frac{e}{2\pi}\times \sqrt{\frac{2}{\alpha(1-\alpha)}} \times t^{-\frac{1}{2}}.
\end{align*}
\end{proof}

\begin{proof}[Proof of Corollary~\ref{cor:LO}]
Note that when $\alpha = \frac{1}{2}$, Lemma \ref{lemma:bernoullicoeff} provides a bound on $\frac{S(n)}{2^n}$. Plug in $\alpha = \frac{1}{2}$ to Lemma \ref{lemma:bernoullicoeff} and combine with Theorem \ref{thm:LO}, we know that if $n > 4$, $C_{LO} = \frac{\sqrt{2}e}{\pi}$ will suffice. If $n \leq 4$, $\frac{2\sqrt{2}e}{\pi} \times n^{-\frac{1}{2}} > 1$ and Lemma \ref{lemma:bernoullicoeff} still holds. 
\end{proof}

\begin{proof}[Proof of Theorem~\ref{thm:A_0}]
We write $|A|$ to denote the cardinality of a finite set A.
\begin{align*}
&\sum_{t \in A_0} |g_{t,i \ominus j}|P(\sum_{\tau=1}^{t-1}\epsilon_\tau g_{\tau,i \ominus j} \geq -\sum_{\tau=1}^{t-1} g_{\tau,i \ominus j} , \sum_{\tau=1}^{t}\epsilon_\tau g_{\tau,i \ominus j} \le -\sum_{\tau=1}^{t} g_{\tau,i \ominus j}) \\
\leq &\sum_{t \in A_0} \frac{1}{\sqrt{T}} \times 1 
= \frac{|A_0|}{\sqrt{T}} 
\leq \sqrt{|A_0|}. \qedhere
\end{align*}
\end{proof}

\begin{proof}[Proof of Theorem~\ref{thm:A_intermediate}]
We write $\epsilon$ with a subset of $[T]$ as subscript to denote $\epsilon_t$'s at times that are within the subset. For example, $\epsilon_{[T]} = \{\epsilon_1, \dots, \epsilon_T\}$. We also write $\epsilon_{-A}$ to denote the set of $\epsilon_t$'s that are within the complement of $A$ with respect to $[T]$.

\noindent
{\bf Case I:} $k \in \{1,\ldots,K-1\}$
\begin{align*}
&\sum_{t \in A_k} |g_{t,i \ominus j}|P\left(\sum_{\tau=1}^{t-1}\epsilon_\tau g_{\tau,i \ominus j} \geq -\sum_{\tau=1}^{t-1} g_{\tau,i \ominus j} , \sum_{\tau=1}^{t}\epsilon_\tau g_{\tau,i \ominus j} \le -\sum_{\tau=1}^{t} g_{\tau,i \ominus j}\right) \\
= &\sum_{t \in A_k} |g_{t,i \ominus j}|\mathbb{E}_{\epsilon_{[T]}}[\mathbbm{1}_{( \sum_{\tau=1}^{t-1}\epsilon_\tau g_{\tau,i \ominus j} \geq -\sum_{\tau=1}^{t-1} g_{\tau,i \ominus j} , \sum_{\tau=1}^{t}\epsilon_\tau g_{\tau,i \ominus j} \le -\sum_{\tau=1}^{t} g_{\tau,i \ominus j} )}] \\
= &\sum_{t \in A_k} |g_{t,i \ominus j}|\mathbb{E}_{\epsilon_{-A_k}}\left[ \mathbb{E}_{\epsilon_{A_k}}[\mathbbm{1}_{( \sum_{\tau=1}^{t-1}\epsilon_\tau g_{\tau,i \ominus j} \geq -\sum_{\tau=1}^{t-1} g_{\tau,i \ominus j} , \sum_{\tau=1}^{t}\epsilon_\tau g_{\tau,i \ominus j} \le -\sum_{\tau=1}^{t} g_{\tau,i \ominus j} )}|\epsilon_{-A_k}]\right] \\
= &\mathbb{E}_{\epsilon_{-A_k}}\left[ \sum_{t \in A_k} |g_{t,i \ominus j}|\mathbb{E}_{\epsilon_{A_k}}[\mathbbm{1}_{( \sum_{\tau=1}^{t-1}\epsilon_\tau g_{\tau,i \ominus j} \geq -\sum_{\tau=1}^{t-1} g_{\tau,i \ominus j} , \sum_{\tau=1}^{t}\epsilon_\tau g_{\tau,i \ominus j} \le -\sum_{\tau=1}^{t} g_{\tau,i \ominus j} )}|\epsilon_{-A_k}]\right] \\
\leq & \sup_{\epsilon_{-A_k}} \sum_{t \in A_k} |g_{t,i \ominus j}|\mathbb{E}_{\epsilon_{A_k}}[\mathbbm{1}_{( \sum_{\tau=1}^{t-1}\epsilon_\tau g_{\tau,i \ominus j} \geq -\sum_{\tau=1}^{t-1} g_{\tau,i \ominus j} , \sum_{\tau=1}^{t}\epsilon_\tau g_{\tau,i \ominus j} \le -\sum_{\tau=1}^{t} g_{\tau,i \ominus j} )}|\epsilon_{-A_k}].
\end{align*}
Let $A_k = \{t_{k,1}, \dots, t_{k,{|A_k|}}\}$ with elements listed in increasing order of time index. Also, define
$$D_n = D_n(\epsilon_{-A_k}) = -\sum_{\tau=1, \tau \in -A_k}^{t_{k,n}-1} (1+\epsilon_{\tau}) g_{\tau,i \ominus j}.
$$
Then, we have
\begin{align*}
&\sum_{t \in A_k} |g_{t,i \ominus j}|\mathbb{E}_{\epsilon_{A_k}}[\mathbbm{1}_{( \sum_{\tau=1}^{t-1}\epsilon_\tau g_{\tau,i \ominus j} \geq -\sum_{\tau=1}^{t-1} g_{\tau,i \ominus j} , \sum_{\tau=1}^{t}\epsilon_\tau g_{\tau,i \ominus j} \le -\sum_{\tau=1}^{t} g_{\tau,i \ominus j} )}|\epsilon_{-A_k}] \\
= &\sum_{n=1}^{|A_k|} |g_{t_{k,n},i \ominus j}|\mathbb{E}_{\epsilon_{A_k}}[\mathbbm{1}_{(\sum_{s=1}^{n-1}\epsilon_{t_{k,s}} g_{t_{k,s},i \ominus j} \geq -\sum_{s=1}^{n-1} g_{t_{k,s},i \ominus j} + D_n , \sum_{s=1}^{n}\epsilon_{t_{k,s}} g_{s,i \ominus j} \le -\sum_{s=1}^{n} g_{t_{k,s},i \ominus j} + D_n)}|\epsilon_{-A_k}] \\
= & \sum_{n=1}^{|A_k|} |g_{t_{k,n},i \ominus j}|P\left(\sum_{s=1}^{n-1}\epsilon_{t_{k,s}} g_{t_{k,s},i \ominus j} \geq -\sum_{s=1}^{n-1} g_{t_{k,s},i \ominus j} + D_n , \right. \\
&\left.
\phantom{\sum_{n=1}^{|A_k|} |g_{t_{k,n},i \ominus j}|P} % for spacing only, won't show up
\sum_{s=1}^{n}\epsilon_{t_{k,s}} g_{t_{k,s},i \ominus j} \le -\sum_{s=1}^{n} g_{t_{k,s},i \ominus j} + D_n \middle |\epsilon_{-A_k}\right).
\end{align*}
By definition of the set $A_k$, we have $|g_{t_{k,s},i \ominus j}| \geq T^{-\frac{1}{2^{k}}}$, so $T^{\frac{1}{2^{k}}} |g_{t_{k,s},i \ominus j}| \geq 1$. Let $M_k = T^{\frac{1}{2^{k}}}$. Then, we have
\begin{align*}
&\sum_{n=1}^{|A_k|} |g_{t_{k,n},i \ominus j}|P\left(\sum_{s=1}^{n-1}\epsilon_{t_{k,s}} g_{t_{k,s},i \ominus j} \geq -\sum_{s=1}^{n-1} g_{t_{k,s},i \ominus j} + D_n , \right. \\
& \left.
\phantom{\sum_{n=1}^{|A_k|} |g_{t_{k,n},i \ominus j}|P} % for spacing only, won't show up
\sum_{s=1}^{n}\epsilon_{t_{k,s}} g_{t_{k,s},i \ominus j} \le -\sum_{s=1}^{n} g_{t_{k,s},i \ominus j} + D_n  \middle |\epsilon_{-A_k} \right) \\
= &\sum_{n=1}^{|A_k|} |g_{t_{k,n},i \ominus j}|P\left(\sum_{s=1}^{n-1}\epsilon_{t_{k,s}} g_{t_{k,s},i \ominus j}M_k \geq -\sum_{s=1}^{n-1} g_{t_{k,s},i \ominus j}M_k + D_n M_k \right., \\
&\left.
\phantom{\sum_{n=1}^{|A_k|} |g_{t_{k,n},i \ominus j}|P} % for spacing only, won't show up
\sum_{s=1}^{n}\epsilon_{t_{k,s}} g_{t_{k,s},i \ominus j}M_k \le -\sum_{s=1}^{n} g_{t_{k,s},i \ominus j}M_k + D_n M_k  \middle |\epsilon_{-A_k} \right) \\
\leq &\sum_{n=1}^{|A_k|} |g_{t_{k,n},i \ominus j}| P\left(\sum_{s=1}^{n}\epsilon_{t_{k,s}} g_{t_{k,s},i \ominus j}M_k \in B_{k,n}  \middle |\epsilon_{-A_k} \right)
\end{align*}
where
$$B_{k,n} = \left[ -\sum_{s=1}^{n} g_{t_{k,s},i \ominus j}M_k + D_nM_k - 2|g_{t_{k,n},i \ominus j}|M_k, -\sum_{s=1}^{n} g_{t_{k,s},i \ominus j}M_k + D_nM_k  \right]$$
is a one-dimensional closed ball with radius $\Delta = |g_{t_{k,n},i \ominus j}|M_k$. Note that this ball is fixed given $\epsilon_{-A_k}$. Since $|g_{t_{k,s},i \ominus j}|M_k \ge 1$, we can apply Corollary \ref{cor:LO} to get 
$$
P\left(\sum_{s=1}^{n}\epsilon_{t_{k,s}} g_{t_{k,s},i \ominus j}M_k \in B_{k,n} \middle |\epsilon_{-A_k} \right) \leq \frac{C_{LO}(\Delta+1)}{\sqrt{n}} = \frac{C_{LO}(|g_{t_{k,n},i \ominus j}|M_k+1)}{\sqrt{n}}.
$$
Now we continue the derivation,
\begin{align*}
&\sum_{n=1}^{|A_k|} |g_{t_{k,n},i \ominus j}| P\left(\sum_{s=1}^{n}\epsilon_{t_{k,s}} g_{t_{k,s},i \ominus j}M_k \in B_{k,n} \middle |\epsilon_{-A_k} \right) \\
\leq & \sum_{n=1}^{|A_k|} |g_{t_{k,n},i \ominus j}|\frac{C_{LO}(|g_{t_{k,n},i \ominus j}|M_k+1)}{\sqrt{n}} \\
\leq & C_{LO}\left(\sum_{n=1}^{|A_k|} \frac{|g_{t_{k,n},i \ominus j}|^2M_k}{\sqrt{n}} +  \sum_{n=1}^{|A_k|} \frac{2}{\sqrt{n}}\right). 
\end{align*}
Since we have $|g_{t_{k,n},i \ominus j}| < T^{-\frac{1}{2^{k+1}}}$, $|g_{t_{k,n},i \ominus j}|^2T^{\frac{1}{2^{k}}} = |g_{t_{k,n},i \ominus j}|^2 M_k < 1$. Thus we have the bound,
\begin{equation*}
C_{LO}\left(\sum_{n=1}^{|A_k|} \frac{|g_{t_{k,n},i \ominus j}|^2M_k}{\sqrt{n}} +  \sum_{n=1}^{|A_k|} \frac{2}{\sqrt{n}}\right) 
\leq  3C_{LO}\sum_{n=1}^{|A_k|} \frac{1}{\sqrt{n}} \leq 6C_{LO}\sqrt{|A_k|}  .
\end{equation*}

\noindent
{\bf Case II:} $k=K$. Similar to the previous case, we have
\begin{align*}
&\sum_{t \in A_{K}} |g_{t,i \ominus j}|P\left(\sum_{\tau=1}^{t-1}\epsilon_\tau g_{\tau,i \ominus j} \geq -\sum_{\tau=1}^{t-1} g_{\tau,i \ominus j} , \sum_{\tau=1}^{t}\epsilon_\tau g_{\tau,i \ominus j} \le -\sum_{\tau=1}^{t} g_{\tau,i \ominus j} \right) \\
\leq & \sup_{\epsilon_{-A_K}} \sum_{t \in A_K} |g_{t,i \ominus j}|\mathbb{E}_{\epsilon_{A_k}}[\mathbbm{1}_{( \sum_{\tau=1}^{t-1}\epsilon_\tau g_{\tau,i \ominus j} \geq -\sum_{\tau=1}^{t-1} g_{\tau,i \ominus j} , \sum_{\tau=1}^{t}\epsilon_\tau g_{\tau,i \ominus j} \le -\sum_{\tau=1}^{t} g_{\tau,i \ominus j} )}|\epsilon_{-A_K}]
\end{align*}
and writing the elements of $A_K$ in increasing order as $\{t_{K,1},\ldots,t_{K,|A_K|}\}$, we get
\begin{align*}
&\sum_{t \in A_K} |g_{t,i \ominus j}|\mathbb{E}_{\epsilon_{A_K}}[\mathbbm{1}_{( \sum_{\tau=1}^{t-1}\epsilon_\tau g_{\tau,i \ominus j} \geq -\sum_{\tau=1}^{t-1} g_{\tau,i \ominus j} , \sum_{\tau=1}^{t}\epsilon_\tau g_{\tau,i \ominus j} \le -\sum_{\tau=1}^{t} g_{\tau,i \ominus j} )}|\epsilon_{-A_K}] \\
\leq &\sum_{n=1}^{|A_K|} |g_{t_{K,n},i \ominus j}| P\left( \sum_{s=1}^{n}\epsilon_{t_{K,s}} g_{t_{K,s},i \ominus j}M_K \in B_{K,n} \right)
\end{align*}
where
$$D_n = D_n(\epsilon_{-A_K}) = -\sum_{\tau=1, \tau \in \{-A_{K}\}}^{t_{K,n}-1} (1+\epsilon_{\tau}) g_{\tau,i \ominus j}  ,$$
$M_K = T^{\frac{1}{2^K}} \leq 2$, and
$$B_{K,n} = \left[-\sum_{s=1}^{n} g_{t_{K,s},i \ominus j}M_K + D_nM_K - 2|g_{t_{K,n},i \ominus j}|M_K, -\sum_{s=1}^{n} g_{t_{K,s},i \ominus j}M_K + D_nM_K\right]$$ is a one-dimensional closed ball with radius $\Delta = |g_{t_{K,n},i \ominus j}|M_K$.
Note that this ball is fixed given $\epsilon_{-A_K}$ and hence, we can apply Corollary \ref{cor:LO} to get
\begin{align*}
&\sum_{n=1}^{|A_K|} |g_{t_{K,n},i \ominus j}| P\left(\sum_{s=1}^{n}\epsilon_{t_{K,s}} g_{t_{K,s},i \ominus j}M_K \in B_{K,n} \right) \\
\leq & \sum_{n=1}^{|A_K|} |g_{t_{K,n},i \ominus j}|\frac{C_{LO}(|g_{t_{K,n},i \ominus j}|M_K+1)}{\sqrt{n}} \\
\leq & C_{LO}(4M_K+2)\sum_{n=1}^{|A_K|} \frac{1}{\sqrt{n}} \leq 20C_{LO}\sqrt{|A_{K}|} .
\end{align*}
Combining the two cases proves the theorem.
\end{proof}

\begin{proof}[Proof of Theorem~\ref{thm:bad_regret}]
Consider a game with two strategies, i.e., $N=2$. We refer to player $i$ as the ``player'' and the other players collectively as the ``environment". On odd rounds, the environment plays payoff vector $(0,0)$. This ensures that after odd rounds, the environment will know exactly which strategy the player will choose as long as there is no tie in the player's sampled cumulative payoffs, because no matter whether the Rademacher random variable is $-1$ or $+1$, the next strategy played will be the same as the strategy the player just played. On even rounds $t$, the environment plays the payoff vector $(0,1-0.1^t)$ if the player chose the first strategy in the previous round, and $(1-0.1^t,0)$ if the player chose the second strategy in the previous round. Under this scenario, we make a critical observation that, as long as the set of sampled time points is not empty, which happens with probability $(\frac{1}{2})^{t-1}$ on round t, there will not be a tie in the cumulative payoffs of the two strategies. Moreover, without a tie, the player will not be able to switch strategy on even rounds so will not accumulate any payoff. Therefore, the total payoff acquired by the player by following sampled fictitious play procedure will be at most $\sum_{t=1}^\infty (\frac{1}{2})^{t-1} = 2$. However, as evident from the environment's procedure, the total payoff for two strategies is at least $0.45T$ and thus the best strategy has a payoff no less than $0.225T$ because of the pigeonhole principle. Hence, the expected regret for the player is at least $0.225T - 2$, which is linear in $T$. 
\end{proof}

\newpage

\begin{center}
\Large{Supplementary Material for ``Sampled Fictitious Play is Hannan Consistent''}
\end{center}

\section{Counterexample Showing Polynomial $N$ Dependence}\label{sec:counterexample}
In this section, we present a counterexample which shows that the sampled fictitious play algorithm ~\eqref{eq:sfp} with Bernoulli sampling~\eqref{eq:bs} has a lower bound of $\Omega(N)$ on its expected regret when $T$ is $2N$.
This is consistent with a lower bound for the expected regret of order $\Omega(\sqrt{NT})$. However, we are unable to extend the construction to an arbitrary $T$.
This counterexample is from \citep{Warmuth15} and we learned it in private communication with Manfred Warmuth and Gergely Neu.
\begin{theorem}[\cite{Warmuth15}]\label{thm:counterexample}
The sampled fictitious play algorithm has expected regret of $\Omega(N)$ when $T$ is $2N$ and $N \rightarrow \infty$.
\end{theorem}
\begin{proof}
Consider the $N \times 2N$ payoff matrix:
\[
\begin{bmatrix}
0  & -1 & -1 & -1 & -1 & -1 & \dots \\
-1 &  0 &  0 & -1 & -1 & -1 & \dots \\
-1 &  0 & -1 &  0 &  0 & -1 & \dots \\
-1 &  0 & -1 &  0 & -1 &  0 & \dots \\
\vdots & \vdots & \vdots & \vdots & \vdots & \vdots & \vdots \\
-1 & 0 & -1 & 0 & -1 & 0 & \dots \\
\end{bmatrix}.
\]
Each row represents a strategy and each column represents payoffs of the strategies in a particular round. In the $m$th odd round, i.e., in the $(2m-1)$-th round, the adversary assigns a payoff of $-1$ to all strategies except strategy $m$ which gets a payoff of $0$. In the $m$th even round, i.e., the $2m$-th round, the adversary assigns a payoff of $-1$ to strategies $1$ through $m$ and a payoff of $0$ to the others. Note that, in all rounds after $2m$, strategy $m$ will always be given a payoff of $-1$. Overall, we will have $N$ strategies and $2N$ rounds, with the best constant strategy being the last strategy which accumulates payoff of $-N$.

To analyze the expected regret, we consider even and odd rounds separately. Note that as long as round $2m-1$ is picked in the sampled history, which happens with probability $\frac{1}{2}$, the algorithm will not choose any strategy from $m+1$ through $N$ at round $2m$. This is because they all have identical payoffs as strategy $m$ prior to round $2m-1$, and strategy $m$ looks better on round $2m-1$. So, the algorithm will pick a strategy from $1$ through $m$ on round $2m$, all of which acquire a gain of $-1$. Therefore, the algorithm will acquire an expected payoff of at most $-\frac{1}{2}$ on even rounds.

Next we consider odd rounds. On round $2m-1$, we observe that the leader set (i.e., the argmin in~\eqref{eq:sfp}) either includes strategy $m$ or not. If it includes strategy $m$, it will additionally include strategies $m+1$ through $N$ as well since they have had identical payoffs in the past. It may possibly also include some strategies in the set $\{1,\dots,m-1\}$. Since the algorithm randomly picks a strategy from the leader set, and all but one of them has a payoff of $-1$ on round $2m-1$, the expected gain of the algorithm is at most $-\frac{N-m}{N-m+1}$. If the leader set does not include strategy $m$, then the expected gain is exactly $-1$ since strategy $m$ is the only one with zero payoff at round $2m-1$. Therefore, the algorithm will acquire an expected payoff of at most $-\frac{N-m}{N-m+1}$ on even rounds.

Hence, the expected regret of Sampled Fictitious Play under this scenario with $N$ strategies and $2N$ rounds is at least, for some $c >0$,
$$
\mathcal{R}_T \ge -N - \Big(-\sum_{m=1}^N (\frac{N-m}{N-m+1}) - \frac{N}{2}\Big)  \ge \frac{N}{2} - c \log(N) = \Omega(N).
$$
\end{proof}

\section{Asymmetric Probabilities}\label{sec:AB}

In this section we prove that for binary payoff and arbitrary $\alpha \in (0,1)$ insead of just $1/2$, the expected regret is $O(\sqrt{T})$ where the constant hidden in $O(\cdot)$ notation blows up in either of the two extreme case:
$\alpha \to 0$ and $\alpha \to 1$. Note that we are still considering the single stream version~\eqref{eq:sfp_single} of the learning procedure. 
\begin{theorem}
For $\alpha \in (0,1)$ and $g_{t} \in \{-1,0,1\}^N$, assuming that $T > \max(\frac{2}{1-\alpha},\frac{2}{\alpha})$, the expected regret satisfies
$$
\mathbb{E}\left[ \mathcal{R}_T \right] \leq \frac{40N^2Q_\alpha}{\alpha} \sqrt{T}
$$
where $Q_\alpha = \frac{e}{2\pi}\times \sqrt{\frac{2}{\alpha(1-\alpha)}}$.
\end{theorem}
\begin{proof}
We begin with the inequality obtained in the proof of Theorem \ref{thm:regret2switches}:
\begin{equation}\label{eq:intermediatebound}
\mathbb{E}\left[ \mathcal{R}_T \right] \leq \frac{N^2}{\alpha} \max\limits_{1 \leq i,j \leq N} \sum_{t=1}^{T}  |g_{t,i \ominus j}|P\left(  \tilde{G}_{t-1,i \ominus j} \geq 0, \tilde{G}_{t,i \ominus j} \le 0 \right). 
\end{equation}
As before, we fix $i$ and $j$, and will bound the expression 
$$
\sum_{t=1}^{T}  |g_{t,i \ominus j}|P( \sum_{\tau=1}^{t-1} (1+\epsilon_\tau)g_{\tau,i \ominus j} \geq 0, \sum_{\tau=1}^{t} (1+\epsilon_\tau)g_{\tau,i \ominus j} \le 0).
$$ 
The rest of the proof is similar to the proof of Theorem \ref{thm:A_intermediate}. Define the classes $A_{k} = \{t: g_{t,i \ominus j} = k, t=1,\dots,T\}$ for $k \in \{-2,-1,1,2\}$.
We have,
\begin{align*}
&\sum_{t=1}^{T}  |g_{t,i \ominus j}|P\left( \sum_{\tau=1}^{t-1} (1+\epsilon_\tau)g_{\tau,i \ominus j} \geq 0, \sum_{\tau=1}^{t} (1+\epsilon_\tau)g_{\tau,i \ominus j} \le 0\right)  \\
\leq& 2\sum_{t=1}^{T} P\left(\sum_{\tau=1}^{t-1}\epsilon_\tau g_{\tau,i \ominus j} \geq -\sum_{\tau=1}^{t-1} g_{\tau,i \ominus j} , \sum_{\tau=1}^{t}\epsilon_\tau g_{\tau,i \ominus j} \le -\sum_{\tau=1}^{t} g_{\tau,i \ominus j}\right) \\
= &2\sum_{k \in \{-2,-1,1,2\}} \sum_{t \in A_k} P\left(\sum_{\tau=1}^{t-1}\epsilon_\tau g_{\tau,i \ominus j} \leq -\sum_{\tau=1}^{t-1} g_{\tau,i \ominus j} , \sum_{\tau=1}^{t}\epsilon_s g_{\tau,i \ominus j} \le -\sum_{\tau=1}^{t} g_{\tau,i \ominus j}\right) .
\end{align*}
For any $k \in \{-2,-1,1,2\}$, 
\begin{align*}
&\sum_{t \in A_k} P\left(\sum_{\tau=1}^{t-1}\epsilon_\tau g_{\tau,i \ominus j} \geq -\sum_{\tau=1}^{t-1} g_{\tau,i \ominus j} , \sum_{\tau=1}^{t}\epsilon_\tau g_{\tau,i \ominus j} \le -\sum_{\tau=1}^{t} g_{\tau,i \ominus j}\right) \\
\leq & \sup_{\epsilon_{-A_k}} \sum_{t \in A_k}\mathbb{E}_{\epsilon_{A_k}}[\mathbbm{1}_{\sum_{\tau=1}^{t-1}\epsilon_\tau g_{\tau,i \ominus j} \geq -\sum_{\tau=1}^{t-1} g_{\tau,i \ominus j} , \sum_{\tau=1}^{t}\epsilon_\tau g_{\tau,i \ominus j} \le -\sum_{\tau=1}^{t} g_{\tau,i \ominus j}}|\epsilon_{-A_k}].
\end{align*}
Let $A_k = \{t_{k,1}, \dots, t_{k,|A_k|}\}$ with elements listed in increasing order of time index. Also define,
for $n \in \{1,\ldots,|A_k|\}$,
$$
D_n =D_n(\epsilon_{-A_k}) =  -\sum_{\tau=1, \tau \in \{-A_k\}}^{t_{k,n}-1}\epsilon_{\tau} g_{\tau,i \ominus j} - \sum_{\tau=1, \tau \in \{-A_k\}}^{t_{k,n}-1}g_{\tau,i \ominus j} .
$$
We then proceed as follows.
\begin{align*}
&\sum_{t \in A_k} \mathbb{E}_{\epsilon_{A_k}}[\mathbbm{1}_{\sum_{\tau=1}^{t-1}\epsilon_\tau g_{\tau,i \ominus j} \geq -\sum_{\tau=1}^{t-1} g_{\tau,i \ominus j} , \sum_{\tau=1}^{t}\epsilon_\tau g_{\tau,i \ominus j} \le -\sum_{\tau=1}^{t} g_{\tau,i \ominus j}}|\epsilon_{-A_k}] \\
= &\sum_{n=1}^{|A_k|} \mathbb{E}_{\epsilon_{A_k}}[\mathbbm{1}_{(\sum_{s=1}^{n-1}\epsilon_{t_{k,s}} g_{t_{k,s},i \ominus j} \geq -\sum_{s=1}^{n-1} g_{t_{k,s},i \ominus j} + D_n , \sum_{s=1}^{n}\epsilon_{t_{k,s}} g_{s,i \ominus j} \le -\sum_{s=1}^{n} g_{t_{k,s},i \ominus j} + D_n)}|\epsilon_{-A_k}] \\
= & \sum_{n=1}^{|A_k|} P\left(\sum_{s=1}^{n-1}\epsilon_{t_{k,s}} g_{t_{k,s},i \ominus j} \geq -\sum_{s=1}^{n-1} g_{t_{k,s},i \ominus j} + D_n, \right. \\
& \phantom{\sum_{n=1}^{|A_k|} P} % for spacing only, won't show up
\left. \sum_{s=1}^{n}\epsilon_{t_{k,s}} g_{t_{k,s},i \ominus j} \le -\sum_{s=1}^{n} g_{t_{k,s},i \ominus j} + D_n \middle| \epsilon_{-A_k} \right) \\
\leq &\sum_{n=1}^{|A_k|} P\left( \bigcup\limits_{u=0}^{4} (\sum_{s=1}^{n}\epsilon_{t_{k,s}} g_{t_{k,s},i \ominus j} = -\sum_{s=1}^{n-1} g_{t_{k,s},i \ominus j} + D_n - u) \middle| \epsilon_{-A_k} \right)  \\
\leq &\sum_{n=1}^{|A_k|}  \left( \sum_{u=0}^4 P \left(\sum_{s=1}^{n}\epsilon_{t_{k,s}} g_{t_{k,s},i \ominus j} = -\sum_{s=1}^{n-1} g_{t_{k,s},i \ominus j} + D_n - u \middle| \epsilon_{-A_k} \right) \right) \\
\leq & 5\sum_{n=1}^{|A_k|} \frac{Q_\alpha}{\sqrt{n}} 
\leq 10Q_\alpha \sqrt{|A_k|} 
\end{align*}
where $Q_\alpha = \frac{e}{2\pi}\times \sqrt{\frac{2}{\alpha(1-\alpha)}}$ from Lemma \ref{lemma:bernoullicoeff}. Putthing things together, we have 
\begin{align*}
\mathbb{E}\left[ \mathcal{R}_T \right] \leq \frac{20N^2Q_\alpha}{\alpha}\sum_{k \in \{-2,-1,1,2\}} \sqrt{|A_k|} \le  \frac{40N^2Q_\alpha}{\alpha}\sqrt{T}.
\end{align*}
\end{proof}

%
%\section{Algorithms}
%\begin{algorithm}
%\caption{Follow the Sampled Leader (Oblivious)}
%\label{alg:FTSL1}
%\begin{algorithmic}[1]
%\Require{Loss vectors $l_{t}$, where $t = 1,\dots, T$}
%\State $t \gets 0$, $\tilde{L}_t \gets 0$,
%\While{$t \leq T$}
%\If{$t = 1$} choose expert $k_1 \in \{1,\dots,N\}$ uniformly at random.
%\Else 
%\State Draw $\epsilon_{t-1}$ from Bernoulli($\alpha$), $\alpha \in (0,1)$.
%\State $\tilde{L}_{t-1} \gets \tilde{L}_{t-2} + (1 + \epsilon_{t-1})l_{t-1}$
%\State choose expert $k_t = \text{arg}\min\limits_i \tilde{L}_{t-1,i}$. 
%Break tie in favor of $k_{t-1}$ \par \hskip\algorithmicindent if $k_{t-1} = \text{arg}\min\limits_i \tilde{L}_{t-1,i}$, otherwise break tie uniformly at random.
%
%\EndIf
%\State $t \gets t + 1$
%\EndWhile
%\end{algorithmic}
%\end{algorithm}
%
%\begin{algorithm}
%\caption{Follow the Sampled Leader (Adaptive)}
%\label{alg:FTSL2}
%\begin{algorithmic}[2]
%\State $t \gets 0$
%\While{$t \leq T$}
%\If{t = 1} choose expert $k_1 \in \{1,\dots,N\}$ uniformly at random.
%\Else 
%\State Draw $\epsilon_1^{(t)},\dots,\epsilon_{t-1}^{(t)}$ from Bernoulli($\frac{1}{2}$) independently. 
%\State $\tilde{L}_{t-1} \gets \sum_{s=1}^{t-1}(1 + \epsilon_{s})l_{s}$
%\State choose expert $k_t = \text{arg}\min \tilde{L}_{t-1,i}$. Break tie in favor of $k_{t-1}$ if \par \hskip\algorithmicindent$k_{t-1} = \text{arg}\min \tilde{L}_{t-1,i}$, otherwise break tie uniformly at random.
%
%\EndIf
%\State $t \gets t + 1$
%\EndWhile
%\end{algorithmic}
%\end{algorithm}

\end{appendices}

\end{document}